\documentclass[11pt]{article}
\usepackage{amssymb,amsmath,amsthm,amscd,latexsym}
\usepackage{mathrsfs}
\usepackage{mathrsfs}
\usepackage{amsfonts}
\usepackage{amsmath}
\usepackage{amssymb}
\usepackage{amscd}
 \usepackage{color}

\renewcommand{\paragraph}{\roman{paragraph}}
\input cyracc.def

\newcommand{\F}{\mathbb{F}}

\newtheorem{thm}{\bfseries  Theorem}[section]
\newtheorem{lem}[thm]{\bfseries   Lemma}

\newtheorem{prop}[thm]{\bfseries   Proposition}
\newtheorem{Def}[thm]{\bfseries  Definition}
\newtheorem{ex}[thm]{\bfseries   Example}
\newtheorem{rem}[thm]{\bfseries   Remark}

\begin{document}
\title{\bf On self-dual negacirculant codes of index two and four
\thanks{This research is supported by National Natural Science Foundation of China (61672036), Technology Foundation for Selected Overseas Chinese Scholar, Ministry of Personnel of China (05015133) and the Open Research Fund of National Mobile Communications Research Laboratory, Southeast University (2015D11) and Key projects of support program for outstanding young talents in Colleges and Universities (gxyqZD2016008).}
}
\author{
\small{Minjia Shi}\\
\small { School of Mathematical Sciences, Anhui University, Hefei, 230601, China}\\
\small{National Mobile Communications Research Laboratory,}\\
 \small{Southeast University, 210096, Nanjing, China }\\
\small{Liqin Qian}\\\small{School of Mathematical Sciences of Anhui University, Hefei, 230601, China }
\and \small{Patrick Sol\'e}\\ \small{CNRS/LAGA, University of Paris 8, 2 rue de la libert\'e, 93 526 Saint-Denis, France}
}

\date{}
\maketitle
\begin{abstract}
We study the asymptotic performance of quasi-twisted codes viewed as modules in the ring $R=\mathbb{F}_q[x]/\langle x^n+1\rangle, $
when they are self-dual and of length $2n$ or $4n.$ In particular, in order for the decomposition to be amenable to analysis, we study
 factorizations of $x^n+1$ over $\mathbb{F}_q, $ with $n$ twice an odd prime, containing only three irreducible factors, all self-reciprocal.
 We give arithmetic conditions bearing on $n$ and $q$ for this to happen. Given a fixed $q,$ we show these conditions are met for infinitely many $n$'s,
provided a refinement of  Artin primitive root conjecture holds.
This number theory conjecture is known to hold under GRH (Generalized Riemann Hypothesis). We derive a modified Varshamov-Gilbert bound on the relative distance of the codes
considered, building on exact enumeration results for given $n$ and $q.$
\end{abstract}
{\bf Keywords:} Negacirculant codes; Self-dual codes; Self-reciprocal polynomials, Artin primitive root conjecture\\
{\bf MSC(2010):} 94 B15, 94 B25, 05 E30

\section{Introduction}

The class of negacirculant codes plays a very significant role in the theory of error-correcting codes as they are a direct generalization of the important family of quasi-cyclic codes,
which is known to be asymptotically good \cite{K1}. They have been used successfully to construct self-dual codes over fields of odd characteristic \cite{1}.
They are a special class of quasi-twisted (QT) codes \cite{J}. A code of length $N$ is {\it quasi-twisted} of {\it index} $l$ for some nonzero scalar $\lambda\in \mathbb{F}_q$,
and of  {\it co-index} $n=N/l$ if it is invariant under the power $T_\lambda^l$ of the constashift $T_\lambda$ defined as
$$T_\lambda(x_0,x_1,\ldots,x_{N-1})\mapsto (\lambda x_{N-1},x_0,\ldots,x_{N-2}).$$ Thus negacirculant codes are QT codes with $\lambda=-1.$

In the present paper, we study the double (resp. four)-negacirculant codes, i.e., index $2$ (resp. $4$) quasi-twisted codes with $\lambda=-1$.
By \cite{J}, QT codes can be decomposed as a direct sum of local rings by the Chinese Remainder Theorem applied to the semilocal ring $R(n,\mathbb{F}_q)=\frac{\mathbb{F}_q[x]}{(x^n-\lambda)}$ . This technique is well suited for studying self-dual QT codes. Although \cite{A1} has already studied self-dual double-negacirculant codes, we consider a different factorization of $x^n+1$ over $\mathbb{F}_q$ to study self-dual double negacirculant codes, and we also study the case of index $4.$ The codes we construct have thus different co-indices, and therefore, different lengths than the codes in \cite{A1}.

In fact, the number of rings occurring in the decomposition of $R(n,\mathbb{F}_q)$ equals the number of irreducible factors of $x^n+1$.
To simplify the analysis, we consider the case when $n=2p$ $(p,q)=1$, $p$ an odd prime $\equiv 3~({\rm mod}~4)$, the alphabet size $q$ is a prime power, and $x^n+1$  can be factored
into a product of three irreducible polynomials,
instead of two for \cite{A1,B1,H}. We demand, to simplify the self-duality conditions, that these factors be self-reciprocal. We give arithmetic conditions bearing on $p,q$ for
these two constraints (three factors and self-reciprocality) on the factorization of $x^n+1$ to hold. These constraints, in turn, are satisfied for infinitely many $p$'s,
provided a refinement of the Artin primitive conjecture holds \cite{M}. When this special factorization is thus enforced, we derive exact enumeration formulae,
and obtain the asymptotic lower bound on the minimum distance of these double (resp. four)-negacirculant codes. This is an analogue of the Varshamov-Gilbert bound \cite{C1,W,K1}.
While this article follows the same line of study as \cite{A1}, it differs in two significant ways: by the value of $n$ (twice a prime instead of a power of $2$), which requires
deep number theoretic facts to establish the factorization of $x^n+1,$ and by the consideration
of four-negacirculant codes, whose structure is more complex than that of double-circulant codes. Note that using the same factorization as in \cite{A1,A2} would have led to using
the euclidean inner product corresponding to reciprocal pairs of polynomials \cite{LS}, and to more difficult equations over finite fields, because
of the primed terms in the generating matrix. Demanding that the factors of $x^n+1$ be self-reciprocal avoids this difficulty, but
requires the theory
of good integers \cite{J1,M0}, and Artin conjecture for arithmetic progressions \cite{M}.

The paper is organized as follows. The next section collects some background material. In Section 3, we present a special kind of factorization of $x^n+1$ over $\mathbb{F}_q$.
Section 4 derives the enumeration formulae of the self-dual double (resp. four)-negacirculant codes of co-index $n$. In Section 5, we study the asymptotic performance of the double
(resp. four)-negacirculant codes. Section 6 is the conclusion of this paper.

\section{Preliminaries}
\subsection{Generator matrix}
If the rows of a matrix over $\mathbb{F}_q,$ with $q$ a prime power, are obtained by successive negashifts from its first row, this matrix is said to be {\it negacirculant}.
Thus, such a matrix is uniquely determined by the polynomial $a(x) \in \mathbb{F}_q[x]$ whose $x$-expansion is the first row of the said matrix.
For example, if $a(x)=a_0+a_1x+\cdots+a_{n-1}x^{n-1}\in \mathbb{F}_q[x]$, then the negacirculant matrix $A$ can be determined as follows:
 $$A=\left(\begin{array}{cccccc}
    a_0 & a_1 & a_2& \cdots & a_{n-2} & a_{n-1} \\
    -a_{n-1} & a_0 & a_1& \cdots & a_{n-3} & a_{n-2} \\
    \vdots & \vdots & \vdots & \ddots & \vdots & \vdots \\
    -a_1& -a_2 & -a_3& \cdots & -a_{n-1} & a_0 \\
  \end{array}
\right).$$

In fact, any $\lambda$-constacyclic code $C$ of length $n$ over $\mathbb{F}_q$ affords a natural structure of module  over the auxiliary ring $R(n,\mathbb{F}_q)=\frac{\mathbb{F}_q[x]}{(x^n-\lambda)}.$ Then, for $\lambda=-1$, double (resp. four)-negacirculant codes are quasi-twisted codes of index $2$ (resp. $4$). And double (resp. four) negacirculant codes are 1 (resp. 2) generator quasi-twisted (QT) codes as a module.

\textbf{Double-negacirculant codes:}

A {\it double-negacirculant} code $C_1$ of length $2n$ is a $[2n,n]$ code over the field $\mathbb{F}_q$ with a generator matrix of the form $$G_1=(I_n~H),$$
 where $I_n$ is the identity matrix of order $n,$ the matrix $H$ is an $n\times n$ negacirculant matrix (i.e., each row is the negashift of the previous one) over $\mathbb{F}_q.$ Let $C_{h}$ be the double-negacirculant codes with the first row of $H$ being the $x$-expansion of $h$ in the ring $R(n,\mathbb{F}_q).$

\textbf{Four-negacirculant codes:}

A {\it four-negacirculant} code $C_2$ of length $4n$ is a $[4n,2n]$ code over the field $\mathbb{F}_q$ with a generator matrix of the form
 $$G_2=\left(
  \begin{array}{cccc}
    I_n & 0 & A & B \\
    0 & I_n & -B^t & A^t \\
  \end{array}
\right),$$
  where $I_n$ is the identity matrix of order $n,$ the matrices $A,B$ are the $n\times n$ negacirculant matrices over $\mathbb{F}_q,$ and the exponent $``t"$ denotes transposition.

We give a general conditions for $G_2$ to generate a self-dual code that generalizes \cite[Proposition 1]{1} from $\mathbb{F}_3$ to $\mathbb{F}_q,$ and from negacirculant matrices to more general
matrices.

{\prop If $G_2$ is as above with the four matrices $A,A^t,B,B^t$ pairwise commuting, and satisfying $AA^t+BB^t+I_n=0$, then $C_2$ is a self-dual code of length $4n,$ and dimension $2n$ over $\mathbb{F}_q$.
In particular, if both $A,$ and $B$ are negacirculant and satisfy $AA^t+BB^t+I_n=0$, then $C_2$ is self-dual.}

\begin{proof}
 Given the hypotheses made on $A,A^t,B,B^t$ it can be checked, using matrix block multiplication,  that $G_2G_2^t=0.$ This shows the self-orthogonality of $C_2.$ Since the matrix
 $G_2$ has rank $2n,$ self-duality  follows.
 In particular, the commutation hypotheses on $A,A^t,B,B^t$ are satisfied if both $A,$ and $B$ are negacirculant. Observe that a negacirculant matrix $A$ with first row polynomial $a(x)$
 is none other that $a(N),$ where $N$ denotes the special negacirculant matrix with attached polynomial $x.$ Thus all pairs of negacirculant matrices of a given order,
 being polynomials in the same matrix $N,$
 commute.
\end{proof}

 In \cite{1}, many self-dual codes with good parameters are constructed in that way.
Let $C_{a,b}$ be the four-negacirculant codes with the first rows of $A,B$ being the $x$-expansion of $a,b$ in the ring $R(n,\mathbb{F}_q).$ Such a four-negacirculant code has a generator matrix $G_2$ with $A,B$ $n\times n$ negacirculant matrices, as described above. Algebraically, we can view such a code $C_{a,b}$ as an $R$-module in $R^4$, generated by $\langle(1,0, a(x),b(x)),(0,1,-b^\prime(x),a^\prime(x))\rangle$, where $a^\prime(x)=a(-x^{n-1})~{\rm mod}~x^n+1$, $b^\prime(x)=b(-x^{n-1})~{\rm mod}~x^n+1$, $R=\mathbb{F}_q[x]/\langle x^n+1\rangle$.
\subsection{Asymptotics}

Now, we recall the rate $r$ and relative distance $\delta$. If $C(n)$ is a family of codes of parameters $[n, k_n, d_n]$, the rate $r$ and relative distance $\delta$ are defined as $$r=\limsup\limits_{n \rightarrow \infty}\frac{k_n}{n},$$ and
\begin{equation*}\delta=\liminf\limits_{n \rightarrow \infty}\frac{d_n}{n}.\end{equation*}
The $q$-ary \emph{entropy function} is defined by \cite{W}:
\begin{equation}H_q(t)=
 \small
\begin{cases}
\emph{ }t{\rm{log}}_q(q-1)-t{\rm{log}}_q(t)-(1-t){\rm{log}}_q(1-t),~0<t<\frac{q-1}{q},\notag \\
 \emph{ }0,~t=0. \\
\end{cases}
\end{equation}
This quantity is instrumental in the estimation of the volume of high-dimensional Hamming balls when the base field is $\mathbb{F}_q$.
The result we are using is that the volume of the Hamming ball of radius $tn$ is, up to subexponential terms, $q^{nH_q(t)}$, when
$0<t<1$, and $n$ goes to infinity \cite[Lemma 2.10.3]{W}.
\section{\textbf{A special factorization of $x^n+1$ over $\mathbb{F}_q$}}
In this section, we first recall the {\it good integers}, a concept introduced by Moree \cite{M0}, with many applications in Coding Theory \cite{J1}. For fixed coprime nonzero integers $l_1$ and $l_2$, a positive integer $s$ is said to be \emph{good} (with respect to $l_1$ and $l_2$) if it is a divisor of $l_1^k+l_2^k$ for some integer $k\geq 1$. Otherwise, $s$ is said to be \emph{bad}. We denote the set of good integers with respect to $l_1$ and $l_2$ by $G_{(l_1,l_2)}$, i.e., $G_{(l_1,l_2)}=\{s\mid s|l_1^k+l_2^k, k\geq 1\}$. We will require the following lemma.
\begin{lem}\label{lem2} (\cite{J1}, Proposition 2.7) Let $l_1,l_2$ and $s>1$ be pairwise coprime odd positive integers and let $\beta\geq 2$ be an integer. Then $2^\beta s\in G_{(l_1,l_2)}$ if and only if ord$_{2^\beta}(\frac{l_1}{l_2})=2$ and $2\parallel$ord$_s(\frac{l_1}{l_2})$. In this case, $2\parallel$ord$_{2^\beta s}(\frac{l_1}{l_2}).$
\end{lem}
We are now ready to describe the main factorization result of this paper.
\begin{thm}\label{tem} Let $n=2p,$ with $ p\equiv 3~({mod}~4),$ an odd prime. Let  $q\equiv 3~({mod}~4),$ denote a prime power, and $(p,q)=1.$ Assume ord$_{4p}(q)=p-1$. Then
  $x^{n}+1$ can be factored into the product of $3$ irreducible polynomials over $\mathbb{F}_q$, namely
  \begin{eqnarray*}
    x^{n}+1 &=& (x^2+1)g_1(x)g_2(x),
  \end{eqnarray*}
 where $deg(g_1(x))=deg(g_2(x))=p-1.$
 \begin{enumerate}
 \item If, on the one hand, the order of $q$ mod $p$ is $\equiv 2 \pmod{4},$ then
 both $g_1(x)$ and $g_2(x)$ are self-reciprocal.
 In particular, if $q$ is primitive modulo $p,$ then we have the said factorization in three factors with both $g_1(x)$ and $g_2(x)$ self-reciprocal.
\item If, on the other hand, the order of $q$ mod $p$ is not $\equiv 2 \pmod{4},$ then
  $g_1(x)$ and $g_2(x)$ are reciprocal of each other.
\end{enumerate}
\end{thm}
\begin{proof} Let $ p\equiv 3~({mod}~4),$ $q\equiv 3~({mod}~4)$ and ord$_{4p}(q)=p-1$.
 If $n=2p, p$ is odd prime, it follows from Theorem 2.45 in \cite{RL} that $x^{2p}+1=\frac{x^{4p}-1}{x^{2p}-1}=(x^2+1)Q_{4p}(x),$ where
 $Q_r()$ denotes the cyclotomic polynomial of order $r.$ As is well-known its degree is $\phi(r),$  {\it Euler's  totient function} that counts the number of integers $m$ with $1\leq m\leq r$ that are relatively prime to $r$. It is easy to check that
 $Q_{4p}(x)$ divides $x^{4p}-1=(x^{2p}-1)(x^{2p}+1),$ but not $x^{2p}-1,$ the said factorization follows.
  Note that, since $q\equiv 3~({\rm mod}~4),$
 $x^2+1$ is an irreducible polynomial over $\mathbb{F}_q$. We compute $\phi(4p)=\phi(4)\phi(p)=2(p-1)$.

 We need $Q_{4p}(x)$ to be factored into $2$ distinct monic irreducible polynomials in $\mathbb{F}_q[x]$ of the same degree $d$. By Theorem 2.47 in \cite{RL}, we obtain $d=\frac{\phi(4p)}{2}=p-1$, where $d$ is
  the multiplicative order of $q$ modulo $4p,$ that is
  the least positive integer $j$ such that $q^j\equiv 1~({\rm mod}~4p).$ Thus ${\rm ord}_{4p}(q)=p-1$.
\begin{enumerate}
\item The following observation shows that the order of $q$ mod $p$ is $\equiv 2 \pmod{4},$ then $g_1(x), g_2(x)$ are self-reciprocal polynomials. If $x^{n}+1=(x^2+1)g_1(x)g_2(x),$ let $\alpha$ be a root of $g_1(x)$, i.e., $g_1(\alpha)=0$, assume that $g_1(x)$ is a self-reciprocal polynomial, then $g_1(\alpha^{-1})=0$, and $\alpha^{-1}\in \{\alpha, \alpha^q, \ldots, \alpha^{q^{deg(g_1(x))-1}}\}$. Thus there exists an integer $i$ such that $\alpha^{-1}=\alpha^{q^i},$ where $0\leq i\leq p-2$ and $i$ is an odd integer. Since $\alpha$ is of order $4p,$ this yields $q^i\equiv -1~({\rm mod}~4p)$. We obtain that $4p\mid q^i+1$, which shows $4p$ is a \emph{good integer} for $l_1=q,$ and $l_2=1.$ According to Lemma \ref{lem2}, we have $4p\in G_{(q, 1)}$ if and only if $2\parallel {\rm ord}_{p}(q)$ ( Note that the condition ord$_4(q)=2$ is equivalent to $q \equiv 3 \pmod{4}$ in this theorem). The necessary condition in Lemma \ref{lem2}, that $2\parallel {\rm ord}_{4p}(q)$ is equivalent to $p \equiv 3 \pmod{4}.$

If $q$ is primitive modulo $p,$ that is if ord$_{p}(q)=p-1$ then we claim that ord$_{4p}(q)=p-1.$ Indeed $q^{p-1}\equiv 1 \pmod{p}.$ Since $q\equiv 3~({\rm mod}~4),$ and $p$ is odd then $q^{p-1}\equiv 1 \pmod{4}.$
By the CRT on integers we have $q^{p-1}\equiv 1 \pmod{4p},$ hence that ord$_p(q)\equiv 2 \pmod{4},$ divides $p-1.$ But ord$_{p}(q)$ divides ord$_{4p}(q)$ by the definition of the order. That gives us ord$_{4p}(q)=p-1$ and the three factor factorization. Since $p\equiv 3~({\rm mod}~4),$ we have ord$_p(q)=p-1\equiv 2 \pmod{4},$ and the condition of self-reciprocity  is also satisfied.
\item  Since $\frac{x^{2p}+1}{x^2+1}$ is self-reciprocal, either $f_1(x)$ and $f_2(x)$ are both self-reciprocal, or they are reciprocal of each other. As the good character of $4p$ is equivalent to the first case of that alternative, by the above observations, the result follows.
\end{enumerate}
\end{proof}
\begin{rem} It is important to observe that ${ord}_{p}(q)$ divides ${ord}_{4p}(q)$ but may be not equal as shows the example $p=11,\, q=27$ when ${ord}_{4p}(q)=10$ but ${ord}_{p}(q)=5.$

\end{rem}
\begin{rem} Let $q\equiv 3~({mod}~4).$ According to the definition of good integers, in fact, $s=4p, l_1=q, l_2=1$ and $k=i$ in Theorem \ref{tem}. That is to say, the condition ``if the order of $q$ mod $p$ is $\equiv 2 \pmod{4}$", ``if there exists an odd integer $i$ such that $q^i\equiv -1~({mod}~4p)$, where $0\leq i\leq p-2$" and ``if $4p$ is a good integer for some integer $0\leq i\leq p-2$ and $i$ is an odd integer" are equivalent. Since $i$ is an odd integer, we say that $4p$ is an \textbf{oddly-good} in the sense of \cite{J1}.
\end{rem}

\begin{ex} Take $n=14$ and $q=3$, then $p=7$ and $q\equiv 3~(mod~4)$. By a simple calculation, we have ord$_{28}(3)=6$ and ord$_7(3)=6\equiv2~({mod}~4)$. From Theorem \ref{tem},.1, we obtain $$x^{14}+1=(x^{2}+1)(x^6+x^5+x^3+x+1)(x^6+2x^5+2x^3+2x+1),$$ there exists an odd integer $i=3$ such that $q^3\equiv -1~({mod}~4p)$.
\end{ex}
\begin{ex} Take $n=22$ and $q=7$, then $p=11$ and $q\equiv 3~(mod~4)$. By a simple calculation, we have ord$_{44}(7)=10$ and ord$_{11}(7)=10\equiv2~({mod}~4)$. From Theorem \ref{tem},.1, we obtain
\begin{eqnarray*}
  x^{22}+1 &=& (x^{2}+1)(x^{10}+2x^9+5x^8+2x^7+6x^6+5x^5+6x^4+2x^3+5x^2+2x+1) \\
   && \cdot(x^{10}+5x^9+5x^8+5x^7+6x^6+2x^5+6x^4+5x^3+5x^2+5x+1),
\end{eqnarray*}
there exists an odd integer $i=5$ such that $q^5\equiv -1~({mod}~4p)$.
\end{ex}
\begin{ex} Take $n=6$ and $q=11$, then $p=3$ and $q\equiv 3~(mod~4)$. By a simple calculation, we have ord$_{12}(11)=2$ and ord$_{3}(11)=2\equiv2~({ mod}~4)$. From Theorem \ref{tem},.1, we obtain $$x^{6}+1=(x^{2}+1)(x^2+5x+1)(x^2+6x+1),$$ there exists an odd integer $i=1$ such that $q\equiv -1~({mod}~4p)$.
\end{ex}
The next examples show what happens in the three factors factorization when $4p$ is not a good integer w.r.t. $q$ and $1.$

\begin{ex} Take $n=22$ and $q=3$, then $p=11$ and $q\equiv 3~(mod~4)$. By a simple calculation, we have ord$_{44}(3)=10$ and ord$_{11}(3)=5\not\equiv2~({mod}~4)$. From Theorem \ref{tem},.2, we obtain $$x^{22}+1=(x^{2}+1)(x^{10}+2x^6+2x^4+2x^2+1)(x^{10}+2x^8+2x^6+2x^4+1),$$ there does not exist an odd integer $i$ such that $q^i\equiv -1~({mod}~4p)$, where $0\leq i\leq p-2$.
\end{ex}

\begin{ex} Take $n=14$ and $q=11$, then $p=7$ and $q\equiv 3~(mod~4)$. By a simple calculation, we have ord$_{28}(11)=6$ and ord$_{7}(11)=3\not\equiv2~({mod}~4)$. From Theorem \ref{tem},.2, we obtain $$x^{14}+1=(x^{2}+1)(x^6+4x^4+6x^2+1)(x^6+6x^4+4x^2+1),$$ there does not exist an odd integer $i$ such that $q^i\equiv -1~({mod}~4p)$, where $0\leq i\leq p-2$.
\end{ex}

\section{\textbf{Exact enumeration}}
 In this section, the self-dual form is due
to the Hermitian inner product. We first recall the definition of the Hermitian inner product.
\begin{Def}
Let $q$ be a prime power. Define the conjugate $\overline{x}$ of $x\in \mathbb{F}_q$ by $\overline{x}=x^{\sqrt{q}}.$ The Hermitian inner product of $\textbf{u}$ and $\textbf{v}$ in $\mathbb{F}_q^n$ is defined by $$\textbf{u}\cdot \textbf{v}=\sum\limits_{i=1}^n u_i\overline{v_i}.$$
\end{Def}
We shall need the following result on a special diagonal equation over a finite field.
\begin{lem}\cite[Cor. 10.1]{A2}\label{a2}
If $q$ is odd, then the number of solutions $(a,b)$ in $\mathbb{F}_{q^2}$ of the equation $a^{1+q}+b^{1+q}=-1$ is $(q+1)(q^2-q)$.
\end{lem}
This result is now used to count certain self-dual QT codes.
\begin{thm}\label{a4}
Let $n=2p,$ with $p$ a prime $\equiv 3~({mod}~4).$ Let $q$ be a prime power satisfying  $q\equiv 3~({mod}~4),$ ord$_{4p}(q)=p-1$ and $(p,q)=1.$ If the order of $q$ mod $p$ is $\equiv 2 \pmod{4},$ then $x^{n}+1$ can be factored into the product of $3$ irreducible polynomials over $\mathbb{F}_q:$ $$x^{n}+1=(x^2+1)g_1(x)g_2(x),$$ where $g_1(x),g_2(x)$ are self-reciprocal polynomials over $\mathbb{F}_q$ with $deg(g_1(x))=deg(g_2(x))=p-1$. If these conditions are satisfied then
\begin{description}
  \item[(1)]  the number of self-dual double-negacirculant codes of length $2n$ over $\mathbb{F}_q$ with generator $\langle (1,h)  \rangle$ is $(q+1)(q^{\frac{p-1}{2}}+1)^2$.
  \item[(2)]  the number of self-dual four-negacirculant codes of length $4n$ over $\mathbb{F}_q$ with $2$-generator $\langle (1,0,a,b), (0,1,-b^\prime, a^\prime)  \rangle$ is $(q+1)(q^2-q)(q^{\frac{p-1}{2}}+1)^2(q^{p-1}-q^{\frac{p-1}{2}})^2$.
\end{description}
\end{thm}
\begin{proof} \textbf{(1)} Let $\mathbb{F}_{q^2}^*=\langle\zeta \rangle$ and $\varepsilon=\zeta^{q-1}$. A factor $x^2+1$ of degree $2$ leads to counting self-dual Hermitian codes of length $2$ over $\mathbb{F}_{q^2}$, that
is to count the solutions of the equation $1+hh^q=0$ over that field. Solving this equation gives the desired result. Furthermore, all the roots of the equation $1+hh^q=0$ are simple and are given by the $1+q$ distinct elements $\zeta^{\frac{q-1}{2}},\varepsilon\zeta^{\frac{q-1}{2}},\ldots,\varepsilon^q\zeta^{\frac{q-1}{2}}$ of $\F_{q^2}$.

Similar to the above discussion, then the number of Hermitian self-dual codes over $\mathbb{F}_{q^{p-1}}$ for the self-reciprocal factor $g_j(x), j=1,2$ of degree $p-1$ is $1+q^{\frac{p-1}{2}}$. Thus the conclusion \textbf{(1)} holds.

\textbf{(2)} In case of a self-reciprocal factor $x^2+1$ of degree $2$, we calculate the number of Hermitian self-dual codes of length $4$ over $\mathbb{F}_{q^2}$. Writing the generator matrix of such a code in the form $\langle (1,0,a_1,b_1), (0,1,-b_1^\prime, a_1^\prime)  \rangle$,
we must count the solutions of the equation:
\begin{equation*}
 \small
\begin{cases}
\emph{ }1+a_1{a_1}^q+b_1{b_1}^q=0, i.e., {a_1}^{1+q}+{b_1}^{1+q}=-1, \\
 \emph{ }1+a_1^\prime{a_1^\prime}^q+b_1^\prime{b_1^\prime}^q=0, i.e., {a_1^\prime}^{1+q}+{b_1^\prime}^{1+q}=-1, \\
 \emph{ }-a_1{b_1^\prime}^q+b_1{a_1^\prime}^q=0. \\
\end{cases}
\end{equation*}
Using the Hermitian scalar product of \cite[Remark 2, p. 2753]{LS}, we see that the prime acts like conjugation $z \mapsto z^q$ over $\mathbb{F}_{q^2}$. Thus $a_1^\prime=a_1^{q}, b_1^\prime=b_1^{q}.$
This shows that the first two equations are equivalent and that the third one is trivial. So
we only need to compute the number of solutions of the equation $a_1^{1+q}+b_1^{1+q}=-1~(*).$ According to Lemma \ref{a2}, then the number of the solutions $(a_1,b_1)$ in $\mathbb{F}_{q^2}$ of $(*)$ is $(q+1)(q^2-q)$.

In case of a self-reciprocal factor $g_1(x)$ of degree $p-1$, we calculate the number of Hermitian self-dual codes of length $4$ over $\mathbb{F}_{q^{p-1}}$. Writing the generator matrix of such a code in the form $\langle (1,0,a_2,b_2), (0,1,-b_2^\prime, a_2^\prime)  \rangle$, it suffices to solve the following system of two equations:
\begin{equation}\label{1}
 \small
\begin{cases}
\emph{ }1+a_2{a_2}^{q^{\frac{p-1}{2}}}+b_2{b_2}^{q^{\frac{p-1}{2}}}=0, i.e., {a_2}^{1+q^{\frac{p-1}{2}}}+{b_2}^{1+q^{\frac{p-1}{2}}}=-1,\\
 \emph{ }1+a_2^\prime{a_2^\prime}^{q^{\frac{p-1}{2}}}+b_2^\prime{b_2^\prime}^{q^{\frac{p-1}{2}}}=0, i.e., {a_2^\prime}^{1+{q^{\frac{p-1}{2}}}}+{b_2^\prime}^{1+{q^{\frac{p-1}{2}}}}=-1.  \\
  \emph{ }-a_2{b_2^\prime}^{q^{\frac{p-1}{2}}}+b_2{a_2^\prime}^{q^{\frac{p-1}{2}}}=0.
\end{cases}
\end{equation}
Using the Hermitian scalar product of \cite{LS}, we see that the prime acts like conjugation $z \mapsto z^{q^{\frac{p-1}{2}}}$ over $\mathbb{F}_{q^{p-1}}$.
Thus $a_2^\prime=a_2^{q^{\frac{p-1}{2}}}, b_2^\prime=b_2^{q^{\frac{p-1}{2}}}.$ This shows that the first two equations are equivalent and that the third one is trivial.
Applying Lemma \ref{a2}, then the number of the solutions $(a_2,b_2)$ in $\mathbb{F}_{q^{p-1}}$ of Equation (\ref{1}) is $(q^{\frac{p-1}{2}}+1)(q^{p-1}-q^{\frac{p-1}{2}})$.

By working in the same way as the case $g_1(x)$, for the self-reciprocal factor $g_2(x)$ of degree $p-1$, we can also deduce that the number of Hermitian self-dual codes over $\mathbb{F}_{q^{p-1}}$ is $(q^{\frac{p-1}{2}}+1)(q^{p-1}-q^{\frac{p-1}{2}})$. This completes the proof of statement \textbf{(2)}.
\end{proof}
\section{\textbf{ Asymptotics of double (resp. four)-negacirculant codes}}
In this section, assume that $q\equiv 3~({{\rm mod}}~4)$ and $p$ is an odd prime with $p\equiv 3~({{\rm mod}}~4), (p,q)=1.$ Set $n=2p.$ As discussed in Section 3, we consider $x^{n}+1=(x^2+1)g_1(x)g_2(x),$ where $g_1(x),g_2(x)$ are self-reciprocal polynomials over $\mathbb{F}_q[x]$ with ${\rm deg}(g_1(x))={\rm deg}(g_2(x))=p-1$. By the Chinese Remainder Theorem (CRT), we have
\begin{eqnarray*}
  \mathbb{F}_q[x]/\langle x^n+1\rangle &\cong& {\mathbb{F}_q[x]}/{\langle x^2+1\rangle}\oplus {\mathbb{F}_q[x]}/{\langle g_1(x)\rangle}\oplus {\mathbb{F}_q[x]}/{\langle g_2(x)\rangle} \\
   &\cong& \mathbb{F}_{q^2} \oplus  \mathbb{F}_{q^{p-1}}\oplus \mathbb{F}_{q^{p-1}}.
\end{eqnarray*}

\begin{lem}\label{lem} If $u\neq0$ has Hamming weight $<2p$,
\begin{description}
  \item[(1)] there are at most $q^2(q^{\frac{n-2}{4}}+1)$ vectors $h$ such that $u\in C_{h}=\langle (1,h)\rangle$ and $C_{h}=C_{h}^\perp$.
  \item[(2)] there are at most $q^4(q^{\frac{n-2}{4}}+1)(q^{\frac{n-2}{2}}-q^{\frac{n-2}{4}})q^{\frac{n-2}{2}}$ pairs $(a,b)$ such that $u\in C_{a,b}=\langle (1,0,a,b), (0,1,-b^\prime, a^\prime)\rangle$ and $C_{a,b}=C_{a,b}^\perp$.
\end{description}

\end{lem}
\begin{proof} \textbf{(1)} Let $C_h=\langle (1, h) \rangle$ and $u=(v, w)$ with $v, w$ vectors of length $n$. The condition of $u\in C_h$ is equivalent to the equations modulo $x^n+1$: $w=vh$. By hypothesis $u\neq 0$ and $wt(u)<2p$, imply $v, w$ cannot both in $(g_1(x)g_2(x))\backslash \{0\}$.

If $v\equiv 0~({\rm mod}~g_1(x)g_2(x))$, then $wt(u)=2p$, it is a contradiction.

If $v\not\in (g_1(x)g_2(x))\backslash \{0\}$, we have the following cases:
\begin{itemize}
\item If $v\equiv 0~({\rm mod}~g_1(x))$, then $1+hh^{q^{\frac{p-1}{2}}}=0$ has $1+q^{\frac{p-1}{2}}$ solutions, and $v\not\equiv 0~({\rm mod}~g_2(x))$, then $h=\frac{w}{v}~{\rm mod}~g_2(x)$ has a unique solution. By the CRT in this case, there are at most $q^2(q^{\frac{n-2}{4}}+1)$ vector $h$ such that $u\in C_{h}=\langle (1,h)\rangle$ and $C_{h}=C_{h}^\perp$. (The $q^2$ comes from the factor $(x^2+1)$).
\item The case $v\not\equiv 0~({\rm mod}~g_1(x))$ and $v\equiv 0~({\rm mod}~g_2(x))$ is similar to the preceding.
\item If $v\not\equiv 0~({\rm mod}~g_1(x))$ and $v\not\equiv 0~({\rm mod}~g_2(x))$, then $h=\frac{w}{v}~{\rm mod}~g_1(x)g_2(x)$ has a unique solution.
\end{itemize}

\textbf{(2)} Let $C_{a,b}=\langle (1,0,a,b), (0,1,-b^\prime, a^\prime)\rangle$ and $u=(c,d,e,f)$ with $c,d, e, f$ vectors of length $n$. The condition of $u\in C_{a,b}$ leads to the following equations modulo $x^n+1$:
\begin{equation}\label{den1}
 \small
\begin{cases}
 \emph{ }e=ca-db^\prime,  \\
   \emph{ }f=cb+da^\prime. \\
\end{cases}
\end{equation}
The hypotheses $u\neq 0$ and $wt(u)<2p$, imply that there cannot have any two of $\{c,d,e,f\}$ both in $(g_1(x)g_2(x))\backslash \{0\}$. Since $c, d$ cannot be both in $(g_1(x)g_2(x))\backslash \{0\}$, without loss of generality, let $d\not\in (g_1(x)g_2(x)).$ It requires to consider three cases:
\begin{enumerate}
  \item[1)] If $d\equiv 0~({\rm mod}~g_1(x))$ and $d\not\equiv 0~({\rm mod}~g_2(x))$;
  \item [2)] If $d\not\equiv 0~({\rm mod}~g_1(x))$ and $d\equiv 0~({\rm mod}~g_2(x))$;
  \item [3)] If $d\not\equiv 0~({\rm mod}~g_1(x))$ and $d\not\equiv 0~({\rm mod}~g_2(x))$.
\end{enumerate}
Next, we will discuss successively these three cases in turn:\\
\hspace*{0.5cm}1) If $d\equiv 0~({\rm mod}~g_1(x))$, we have \begin{equation*}\label{den1}
 \small
\begin{cases}
 \emph{ }e=ca~{\rm mod}~g_1(x),  \\
   \emph{ }f=cb~{\rm mod}~g_1(x). \\
\end{cases}
\end{equation*}
\begin{itemize}
  \item a) If $c\not\equiv 0~({\rm mod}~g_1(x))$, then there exists unique pair $(a,b)~({\rm mod}~g_1(x))$.
  \item b) If $c\equiv 0~({\rm mod}~g_1(x))$, then $c,d,e,f\equiv 0~({\rm mod}~g_1(x))$, and the above system does not bring any information on $a$ and $b.$
\end{itemize}
But we know that $C_{a,b}=C_{a,b}^\perp$, and thus ${a}^{1+q^{\frac{p-1}{2}}}+{b}^{1+q^{\frac{p-1}{2}}}=-1.$ Therefore, there are at most $(q^{\frac{p-1}{2}}+1)(q^{p-1}-q^{\frac{p-1}{2}})$ pairs $(a,b)~({\rm mod}~g_1(x))$.
And because $d\not\equiv 0~({\rm mod}~g_2(x))$, we obtain $b^\prime=\frac{ca-e}{d}~{\rm mod}~g_2(x).$ Thus, for a given $a$, $b$ is unique. Note that there are $q^{p-1}$ choices for $a$.

All the above show that, by the CRT in case 1), there are at most $q^4(q^{\frac{p-1}{2}}+1)(q^{p-1}-q^{\frac{p-1}{2}})q^{p-1}$ pairs $(a,b)$ such that $u\in C_{a,b}=\langle (1,0,a,b), (0,1,-b^\prime, a^\prime)\rangle$ and $C_{a,b}=C_{a,b}^\perp$ (The $q^4$ comes from the factor $(x^2+1)$).

2) This case is symmetric to 1) in the discussion, hence it is omitted.

3) In this case, for given $a$, there exists unique $b~({\rm mod}~g_1(x)g_2(x))$. There are $q^n$ choices for $a$, i.e., there are $q^{n}$ pairs $(a,b)$. Note that this is dominated by the count in $1),$
which is of order $O(q^{\frac{5n}{4}}).$

Hence, combining these three scenarios, we obtain that there are at most $q^4(q^{\frac{n-2}{4}}+1)(q^{\frac{n-2}{2}}-q^{\frac{n-2}{4}})q^{\frac{n-2}{2}}$ pairs $(a,b)$ such that $u\in C_{a,b}=\langle (1,0,a,b), (0,1,-b^\prime, a^\prime)\rangle$ and $C_{a,b}=C_{a,b}^\perp$.
\end{proof}

We are now ready to present the main results in this section.
\begin{thm} Under the hypotheses of Theorem 3.2, if $q$ is prime power and $q\equiv 3~({mod}~4),$ $n=2p$ and $(p,q)=1$, then
\begin{description}
  \item[(1)] there are infinite families of self-dual double-negacirculant codes of length $2n$ over $\F_q,$ of relative distance $\delta,$ satisfying $H_q(\delta)\geq\frac{1}{8}$.
  \item[(2)] there are infinite families of self-dual four-negacirculant codes of length $4n$ over $\F_q,$ of relative distance $\delta,$ satisfying $H_q(\delta)\geq\frac{1}{16}$.
\end{description}
\end{thm}
\begin{proof} Observe first that the infinitude of the primes $p\equiv 3~({\rm mod}~4),$ admitting a fixed number $q$ as a primitive root is guaranteed by \cite[Theorem 3]{M}, which derives (under GRH) a density for this family of primes. We now derive statements (1) and (2) by the classical
technique of ``expurgated random coding", which shows the existence of codes with distance
at least a given bound in a family of codes by showing that the codes with distance below
that bound are less in number than the size of the family.

 \textbf{(1)} Denote by $d_n$ the largest integer satisfying
 $$ q^2(q^{\frac{n-2}{4}}+1)\sum_{i=0}^{d_n}{n \choose i}(q-1)^i < q^{n/2}.  $$
 By Lemma 5.1 (i) there are codes of length $2n$ in the family we consider of distance $\ge d_n.$ Letting $d_n \sim n\delta_0,$ by the entropic estimates of \cite[Lemma 2.10.3]{W},
 we see that the defining inequality for $d_n$ will be satisfied for large $n$ if $H_q(\delta_0)=\frac{1}{8}.$ We see that the family of codes of length $n$ thus constructed with minimum distance $\ge d_n$ has relative distance $\delta\ge \delta_0.$

\textbf{(2)} The argument is similar to \textbf{(1)} and is omitted.
\end{proof}

{\rem The rates of the codes in this result are both $1/2.$ Thus our bounds are below the Varshamov-Gilbert bound for the linear codes which would be $H_q(\delta)\geq\frac{1}{2}.$}

\section{\textbf{Conclusion}}
In this paper, we have studied the class of self-dual double (resp. four)-negacirculant codes over finite fields. Motivated by \cite{A1}, A. Alahmadi et al. have studied the special factorization $x^n+1$ when it factors into two irreducible factors reciprocal of each other for $n$ a power of $2$. Here, we have studied another special kind of decomposition, for $n$ twice an odd prime, when $x^n+1$ factors into three irreducible and self-reciprocal polynomials. Further, we have derived an exact enumeration formula for this family of self-dual double (resp. four)-negacirculant codes. The modified Varshamov-Gilbert bound on the relative minimum distance we derived relies on some deep number-theoretic conjectures (Artin primitive root in arithmetic progression or Generalized Riemann Hypothesis).
 It would be a worthwhile study to consider more general factorizations of $x^n+1,$ and also to look at QT codes with more than two generators.

\end{document}